\documentclass[letterpaper, 10 pt, conference]{ieeeconf}  

\IEEEoverridecommandlockouts                              
\usepackage{graphics} 
\usepackage{epsfig} 
\usepackage{amsmath} 
\usepackage{amssymb}  
\usepackage{xcolor}
\usepackage{graphicx}

\newtheorem{theorem}{Theorem}
\newtheorem{lemma}{Lemma}
\newtheorem{definition}{Definition}

\newtheorem{corollary}{Corollary}
\newtheorem{remark}{Remark}

\title{\LARGE \bf
	Macroscopic Motion Patterns from Generalized Velocity Rigidity in Multi-Agent Networks
}

\author{Ronghai He and Changhuang Wan% 
	\thanks{The authors are with the Department of Mechanical Engineering, City University of Hong Kong, Hong Kong, China.
		{\tt\small ronghaihe2-c@my.cityu.edu.hk, changwan@cityu.edu.hk}}%
}

\begin{document}
	
	\maketitle
	\thispagestyle{empty}
	\pagestyle{empty}
	
	%%%%%%%%%%%%%%%%%%%%%%%%%%%%%%%%%%%%%%%%%%%%%%%%%%%%%%%%%%%%%%%%%%%%%%%%%%%%%%%%
	\begin{abstract}
		Traditional graph rigidity theory enforces structural constraints in the position space, effectively restricting a multi-agent formation to a static geometric shape. In this paper, we introduce a fundamental paradigm shift by applying rigidity constraints directly to the agents' continuous-time velocity space. We propose the concept of generalized velocity rigidity, demonstrating that the macroscopic physical motion patterns of a multi-agent network are entirely dictated by the underlying static graph topology. By strictly confining the network's acceleration profile to the null space of the velocity rigidity matrix, we mathematically map the trivial motions of this null space into exact physical trajectories. Specifically, we prove that translational, rotational, and scaling trivial motions seamlessly integrate into cohesive curved flocking, synchronized helical and circular $\mathrm{SO}(d)$ orbiting, and exponential spatial homothety, respectively. Furthermore, we analyze velocity-space singularities, showing that velocity consensus induces a valid topological contraction rather than a structural failure. Finally, we provide comprehensive numerical simulations to validate these theoretical mappings, demonstrating that complex macroscopic maneuvers can be orchestrated purely through algebraic constraints in the velocity space, unlocking unprecedented spatial shape flexibility for multi-agent swarms.
	\end{abstract}
	\section{Introduction}
	The coordinated control of multi-agent systems relies heavily on the structural properties of the underlying communication and sensing topologies \cite{anderson2018rigid, olfati2007consensus, oh2015survey}. Over the past decades, graph rigidity theory has emerged as a fundamental mathematical tool for distributed formation control \cite{oh2015survey, krick2009stabilisation, deQueiroz2019formation}, network localization \cite{aspnes2006theory, eren2004rigidity}, and target enclosure \cite{lan2011}. Classical rigidity frameworks—encompassing distance \cite{cortes2004coverage, krick2009stabilisation, sun2014finite}, bearing \cite{zelazo2019formation}, and angle constraints \cite{huang2025signed, chen2020angle}—operate predominantly in the position space. The primary objective of these conventional approaches is to enforce specific inter-agent geometric conditions to maintain a static, global formation shape up to a set of allowable, non-deforming trivial motions \cite{anderson2018rigid, deQueiroz2019formation}. For instance, distance rigidity rigidly preserves the exact scale of the formation \cite{krick2009stabilisation, cai2014rigidity}, whereas local bearing \cite{zelazo2019formation} and angle-based \cite{huang2025signed, chen2020angle} rigidity inherently permit uniform spatial scaling, providing a limited degree of size adaptability while preserving the fundamental geometric shape \cite{zelazo2019formation}.
	
	Despite its profound success, position-based rigidity inherently restricts a multi-agent network to act as a highly constrained structural body \cite{krick2009stabilisation, deQueiroz2019formation, cai2014rigidity}. In complex, highly dynamic, or obstacle-dense environments, locking the absolute relative positions limits the fluidity of the swarm \cite{olfati2006flocking, su2009connectivity}. To introduce shape flexibility, recent literature has intensively investigated affine formation maneuvers \cite{zhao2018affine, xu2020affine, lin2015necessary, ma2020event,  kang2024affine}. Affine rigidity allows a network to undergo global affine transformations—such as shearing, non-uniform scaling, and rotational maneuvers—by manipulating a nominal spatial configuration \cite{zhao2018affine, xu2020affine}.
	
	In a parallel thread of research addressing dynamic environments, collective motion and flocking have been extensively studied, pioneering from Reynolds' foundational heuristic boids model \cite{reynolds1987flocks} to rigorous system-theoretic analyses of consensus-based cohesive flocking \cite{olfati2007consensus, olfati2006flocking, jadbabaie2003coordination,  cucker2007emergent}. Recent advances in flocking encompass adaptive potential functions \cite{li2024flocking}, data-driven barrier functions \cite{yan2025data}, fixed-time optimal control \cite{liu2025distributed}, predictive lattice formations \cite{hernandez2026predictive}, and cooperative-competitive dynamics \cite{wang2026asynchronous,ali2024state}. While these flocking models elegantly achieve velocity alignment and group cohesion, they predominantly rely on gradient-based artificial potential fields or nearest-neighbor averaging \cite{olfati2006flocking, su2009connectivity, li2024flocking}, rather than utilizing a strict algebraic graph structure to uniquely dictate the resulting macroscopic group geometry \cite{anderson2018rigid, krick2009stabilisation}. Consequently, both affine shape-shifting \cite{zhao2018affine, xu2020affine} and traditional/recent flocking protocols \cite{olfati2006flocking, cucker2007emergent, li2024flocking, ali2024state} lack a unified framework that directly maps static topological constraints to dynamic, macroscopic group trajectories.
	
	In this paper, we propose a novel theoretical framework: \textbf{Generalized Velocity Rigidity}. By migrating rigidity constraints from the position space to the velocity space, we redefine the structural integrity of the network. A framework is velocity-rigid if its allowable continuous-time infinitesimal motions—which physically correspond to the agents' acceleration profiles—are strictly confined to the null space of the velocity rigidity matrix. This single, elegant differential inclusion serves as a powerful bridge between algebraic graph theory and macroscopic physical kinematics, effectively unifying the structural guarantees of rigidity theory with the dynamic fluidity of flocking algorithms.
	
	The core contribution of this work is the rigorous mathematical demonstration that the dynamic macroscopic trajectories of a multi-agent system are completely and uniquely dictated by its static algebraic graph topology in the velocity space. Specifically, this paper makes the following key contributions:
	\begin{enumerate}
		\item \textbf{Geometric Configurations and Singularities}: We demonstrate that degenerate states (e.g., velocity consensus or collinearity) correspond to physically meaningful topological contractions and subspace confinements rather than physical failures.
		\item \textbf{Exact Kinematic Mappings}: We establish analytical proofs mapping the algebraic null space of the velocity rigidity matrix to physical trajectories. We demonstrate that translational, rotational, and scaling trivial motions in the velocity graph strictly map onto curved flocking, concentric $SO(d)$ Lie group orbiting, and exponential spatial dilation in the Euclidean position space, respectively.
		\item \textbf{Decoupling of Topology and Spatial Geometry}: Through rigorous theoretical analysis and comprehensive numerical case studies, we demonstrate a fundamental advantage of generalized velocity rigidity: it rigidly locks the macroscopic continuous dynamics while inherently permitting extreme spatial shape flexibility (e.g., continuous stretching and shearing). We explicitly identify the algebraic position-velocity coupling conditions required to enforce classical shape-preserving maneuvers (such as concentric orbiting and spatial homothety).
	\end{enumerate}
	
	The remainder of this paper is organized as follows. Section \ref{sec:prelim} introduces the mathematical preliminaries and formalizes generalized velocity rigidity. Section \ref{sec:KINEMATIC_MAPPING} provides the core theorems and proofs for the exact kinematic mappings from the null space to macroscopic patterns. Section \ref{sec:DEGENERATE} investigates the geometric configurations and singularity properties in the velocity space. Section \ref{sec:simulations} provides numerical simulations to validate the theoretical findings. Finally, Section \ref{sec:conclusion} concludes the paper and discusses future research directions.
	
	% =========================================================
	\section{PRELIMINARIES AND GENERALIZED VELOCITY RIGIDITY}\label{sec:prelim}
	% =========================================================
	Let $\mathbb{R}^d$ denote the $d$-dimensional Euclidean space ($d \ge 2$). The Kronecker product is denoted by $\otimes$, and $I_n$ represents the $n \times n$ identity matrix.
	
	\subsection{Algebraic Graph Theory and Velocity Frameworks}
	Consider a multi-agent system of $n$ agents modeled by a graph $\mathcal{G} = (\mathcal{V}, \mathcal{E})$, where $\mathcal{V} = \{1, \dots, n\}$ is the set of vertices and $\mathcal{E} \subseteq \mathcal{V} \times \mathcal{V}$ is the set of $m$ edges representing measurement interactions. Unlike classical rigidity theory, which assigns a position vector $p_i \in \mathbb{R}^d$ to each vertex {to represent relative distance/bearing/angle constraints}, we project the network into the continuous-time velocity space. Let $v_i(t) \in \mathbb{R}^d$ and $a_i(t) \in \mathbb{R}^d$ denote the velocity and acceleration of agent $i$, respectively. The stacked velocity vector is $v(t) = [v_1(t)^T, \dots, v_n(t)^T]^T \in \mathbb{R}^{dn}$, and the acceleration vector $a(t) = [a_1(t)^T, \dots, a_n(t)^T]^T \in \mathbb{R}^{dn}$. 
	
	\begin{definition}[\textbf{Velocity Framework}]
		A velocity framework, denoted as $(\mathcal{G}, v)$, is a pair consisting of a graph $\mathcal{G}$ and a specific velocity realization $v \in \mathbb{R}^{dn}$ assigned to its vertices.
	\end{definition}
	
	\begin{definition}[\textbf{Velocity Generic Configuration}]
		A velocity framework $(\mathcal{G}, v)$ is generic if the coordinates of $v \in \mathbb{R}^{dn}$ are algebraically independent over the field of rational numbers. For a generic framework, the structural matrices achieve their maximal possible rank.
	\end{definition}
	
	\subsection{The Generalized Velocity Rigidity Matrix}
	For any edge $k = (i,j) \in \mathcal{E}$, let $r_k(v_i, v_j)$ represent a measurement constraint acting on the relative velocity of connected agents. This accommodates various modalities; for example, distance constraints yield $r_k = \|v_i - v_j\|^2$, while bearing constraints yield $r_k = \frac{v_i - v_j}{\|v_i - v_j\|}$. Stacking the functions yields the generalized velocity rigidity function, $r_{\mathcal{G}}(v) : \mathbb{R}^{dn} \to \mathbb{R}^{m_c}$. 
	
	\begin{definition}[\textbf{Velocity Rigidity Matrix}]
		The generalized velocity rigidity matrix, $R_v(v) \in \mathbb{R}^{m_c \times dn}$, is the partial derivative of the generalized velocity rigidity function with respect to the stacked velocity vector:
		\begin{equation}
			R_v(v) = \frac{\partial r_{\mathcal{G}}(v)}{\partial v}.
		\end{equation}
	\end{definition}
	
	\subsection{Infinitesimal Rigidity and Continuous-Time Kinematics}
	An infinitesimal perturbation $\delta v \in \mathbb{R}^{dn}$ preserves the velocity constraints to a first-order approximation if $R_v(v)\delta v = \mathbf{0}$. Any graph subjected to relative constraints inherently allows a set of unconstrained, non-deforming aggregate movements, defined as the space of trivial motions $\mathcal{S}_{trivial}$.
	
	\begin{definition}[\textbf{Infinitesimal Velocity Rigidity}]
		A framework $(\mathcal{G}, v)$ is infinitesimally velocity rigid if every valid infinitesimal perturbation exclusively corresponds to a trivial motion, requiring:
		\begin{equation}
			\text{Null}(R_v(v)) = \mathcal{S}_{trivial}.
		\end{equation}
	\end{definition}
	
	The critical departure from classical statics lies in continuous-time kinematics. By differentiating the velocity, we obtain the acceleration profile $a(t) = \dot{v}(t) \in \mathbb{R}^{dn}$.
	
	\begin{remark}[\textbf{Structural Isomorphism}]
		It is crucial to note that the generalized velocity rigidity matrix $R_v(v)$ exhibits a profound structural isomorphism with the classical position-based rigidity matrix $R(p)$. For instance, under distance-based constraints, the Jacobian blocks $\frac{\partial \|v_i-v_j\|^2}{\partial v_i} = 2(v_i-v_j)^T$ mathematically mirror their positional counterparts. Consequently, the well-established topological conditions for minimal infinitesimal rigidity (e.g., the Laman condition requiring $m = 2n-3$ edges in 2D \cite{zhao2017laman}) are directly inheritable to the velocity space, provided the generic configuration assumption holds.
	\end{remark}
	
	\begin{lemma} \label{lem-diff-inclusion}[\textbf{Kinematic Differential Inclusion}]
		For a network to maintain infinitesimal velocity rigidity continuously, its global acceleration vector must be strictly confined to the instantaneous null space:
		\begin{equation}
			\label{eq:diff_inclusion}
			a(t) \in \text{Null}(R_v(v(t))).
		\end{equation}
	\end{lemma}
	
	This differential inclusion forms the theoretical bedrock of this paper. By mapping the basis vectors of $\text{Null}(R_v)$ through integration, we mathematically deduce the macroscopic spatial trajectories a velocity-rigid swarm is permitted to execute.

	%=======================================================
	\section{KINEMATIC MAPPING FROM NULL SPACE TO MACROSCOPIC MOTION PATTERNS}\label{sec:KINEMATIC_MAPPING}
	We establish the exact analytical mapping from the algebraic null space of $R_v(v)$ to the macroscopic physical trajectories $p(t)$. We assume the framework $(\mathcal{G}, v)$ remains infinitesimally velocity rigid: $a(t) \in \text{Null}(R_v(v(t)))$. Let $p_i(t), v_i(t), a_i(t) \in \mathbb{R}^d$ denote position, velocity, and acceleration, such that $\dot{p}_i(t) = v_i(t)$ and $\dot{v}_i(t) = a_i(t)$.
	
	\subsection{Translational Subspace to Curved Flocking}
	\begin{lemma}[\textbf{Translational Invariance}]
		For any generalized velocity rigidity matrix $R_v(v)$ derived strictly from relative measurements, the translational subspace $\mathcal{S}_T = \text{span}\{\mathbf{1}_n \otimes I_d\}$ is a guaranteed subspace of $\text{Null}(R_v(v))$.
	\end{lemma}
	\begin{proof}
		Relative measurements depend exclusively on velocity differences. Applying a uniform shift $\delta v = \mathbf{1}_n \otimes \mathbf{c}$ yields $(v_i + \mathbf{c}) - (v_j + \mathbf{c}) = v_i - v_j$. Since measurements are invariant, the directional derivative along $\delta v$ is zero, hence $R_v(v)(\mathbf{1}_n \otimes \mathbf{c}) = \mathbf{0}$. % \hfill $\blacksquare$
	\end{proof}
	
	\begin{theorem}[\textbf{Curved Flocking}]
		\label{thm:Translational}
		If a velocity-rigid framework's acceleration profile is purely confined to the translational null space ($a(t) \in \mathcal{S}_T$), the macroscopic motion is a cohesive curved flocking maneuver, and velocity rigidity constraints are preserved globally over time.
	\end{theorem}
	\begin{proof}
		Since $a(t) \in \mathcal{S}_T$, the acceleration is identical for all agents: $a_i(t) = c(t)$. Integrating with respect to time yields:
		\begin{equation}
			v_i(t) = v_i(0) + \int_{0}^{t} c(\tau) \text{d}\tau.
		\end{equation}
		For any edge $(i,j) \in \mathcal{E}$, the relative velocity is $v_i(t) - v_j(t) = v_i(0) - v_j(0)$, proving that exact global velocity rigidity is maintained. Integrating the velocity yields the spatial trajectory:
		\begin{equation}
			p_i(t) = p_i(0) + v_i(0)t + \int_{0}^{t} \int_{0}^{\tau} {c}(s) \text{d}s \text{d}\tau.
		\end{equation}
		This demonstrates that all agents undergo a synchronized translation superimposed onto their initial linear trajectories, manifesting as a macroscopic curved flocking pattern. % \hfill $\blacksquare$
	\end{proof}
	
	\begin{figure}[htbp]
		\centering
		\includegraphics[width=\linewidth]{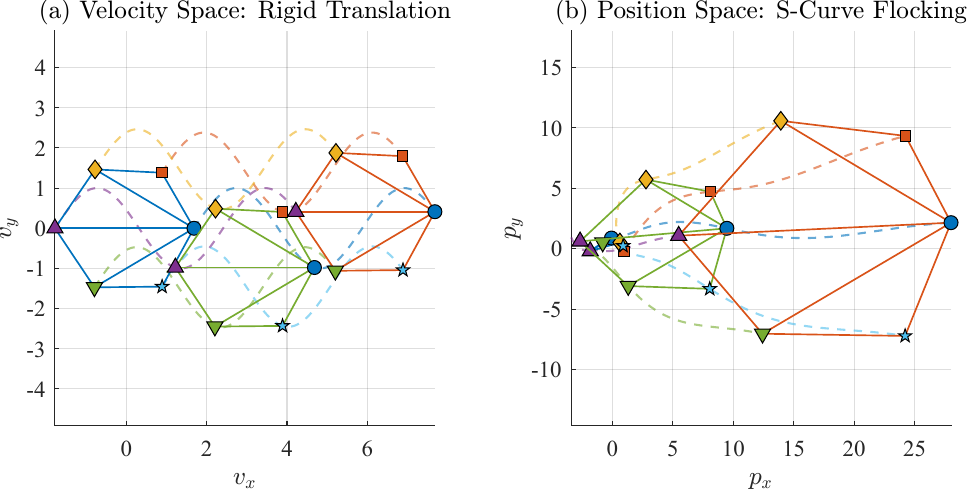}
		\vspace{-5mm}
		\caption{Kinematic mapping from $\mathcal{S}_T$ to curved flocking. (a) In the velocity space, the framework translates cohesively driven by a time-varying consensus acceleration ${c}(t)$. (b) The integration generates synchronized curved trajectories.}
		\label{fig:translation2}
		\vspace{-3mm}
	\end{figure}
	
	Fig. \ref{fig:translation2} visually demonstrates this exact mapping. The consensus acceleration ${c}(t)$ translates the entire velocity framework along a curved path. The transition of the solid framework confirms that the structural topology is rigorously preserved without rotation or scaling. In the position space, this velocity shift integrates into synchronized curved flocking trajectories.
	
	\subsection{Rotational Subspace to Helical and Circular Orbiting}
	\begin{lemma}[\textbf{Rotational Subspace}]
		For a velocity graph utilizing distance or angle measurements, the rotational subspace $\mathcal{S}_R = \{ (I_n \otimes \Omega)v \mid \Omega \in \mathfrak{so}(d) \}$ is a subset of $\text{Null}(R_v(v))$.
	\end{lemma}
	\begin{proof}
		The proof follows from the invariance of lengths and angles under $SO(d)$. A global rotation in velocity space yields $\tilde{v}_i = Q v_i$. The infinitesimal generator is $Q \approx I_d + \epsilon \Omega$, where $\Omega = -\Omega^T \in \mathfrak{so}(d)$. The trivial motion is $\delta v_i = \Omega v_i$, stacking to $\delta v = (I_n \otimes \Omega)v$. Since distance/angle functions are invariant to $SO(d)$, their Jacobians annihilate these generators. % \hfill $\blacksquare$
	\end{proof}
	
	\begin{theorem}[\textbf{Synchronized Helical and Circular Orbiting}]\label{thm:Rotational}
		Assume the framework is velocity-rigid and its acceleration is strictly confined to a time-invariant rotational null space $a(t) = (I_n \otimes \Omega)v(t)$ for $\Omega \in \mathfrak{so}(d)$. Then, the velocity vectors map onto a Lie group manifold, preserving the graph's topology. Concurrently, the macroscopic integration maps the agents onto synchronized \textbf{helical trajectories} in the position space. These degenerate into \textbf{circular orbits} if and only if the initial velocities reside entirely within the range space of $\Omega$.
	\end{theorem}
	\begin{proof}
		The prescribed acceleration yields an ordinary differential equation:
		\begin{equation}
			\dot{v}_i(t) = \Omega v_i(t).
		\end{equation}
		The unique analytical solution is $v_i(t) = e^{\Omega t} v_i(0)$. Since $\Omega \in \mathfrak{so}(d)$, the matrix exponential $e^{\Omega t} \in SO(d)$ constitutes a global rotation, ensuring $\|v_i(t) - v_j(t)\|^2 = \|v_i(0) - v_j(0)\|^2$. Exact velocity rigidity is preserved.
		
		To determine the macroscopic position trajectory ($d \ge 2$), we orthogonally decompose the initial velocity as $v_i(0) = v_i^\perp + v_i^\parallel$, where $v_i^\perp \in \text{Range}(\Omega)$ and $v_i^\parallel \in \text{Null}(\Omega)$. Since $\Omega v_i^\parallel = \mathbf{0}$, $e^{\Omega t}v_i^\parallel = v_i^\parallel$. Integrating the velocity yields:
		\begin{align}
			p_i(t) &= p_i(0) + \int_{0}^{t} e^{\Omega \tau} (v_i^\perp + v_i^\parallel) \text{d}\tau \nonumber\\
			&= p_i(0) + \int_{0}^{t} e^{\Omega \tau} v_i^\perp \text{d}\tau + v_i^\parallel t.
		\end{align}
		Let $\Omega^\dagger$ denote the Moore-Penrose pseudoinverse. Because $v_i^\perp$ is strictly in the range of $\Omega$, the integral evaluated on this orthogonal complement yields a bounded periodic function:
		\begin{equation}
			p_i(t) = \underbrace{p_i(0) + \Omega^\dagger(e^{\Omega t} - I_d)v_i^\perp}_{\text{Circular orbiting}} + \underbrace{v_i^\parallel t}_{\text{Linear drift}}.
		\end{equation}
		Geometrically, the agent traces a circular orbit around an individualized moving axis $c_i(t) = p_i(0) - \Omega^\dagger v_i^\perp + v_i^\parallel t$. If and only if $v_i(0) \in \text{Range}(\Omega)$, the null space projection vanishes ($v_i^\parallel = \mathbf{0}$). The linear drift is eliminated, generating circular orbits revolving around static individual centers $c_i = p_i(0) - \Omega^\dagger v_i(0)$. % \hfill $\blacksquare$
	\end{proof}
	
	\begin{remark}
		The dimensionality $d$ dictates the invertibility of $\Omega$. When $d=3$, $\Omega$ is universally singular, inherently yielding helical trajectories unless the initial velocity is strictly orthogonal to the rotational axis. Conversely, for generic 2D planar networks, $\Omega$ is strictly of full rank ($\text{Null}(\Omega) = \{\mathbf{0}\}$), guaranteeing $v_i^\parallel = \mathbf{0}$, meaning any 2D rotational trivial motion generates pure circular orbits.
	\end{remark}
	
	\vspace{-3mm}
	\begin{figure}[htbp]
		\centering
		\hspace{-10mm}
		\includegraphics[width=1.1\linewidth]{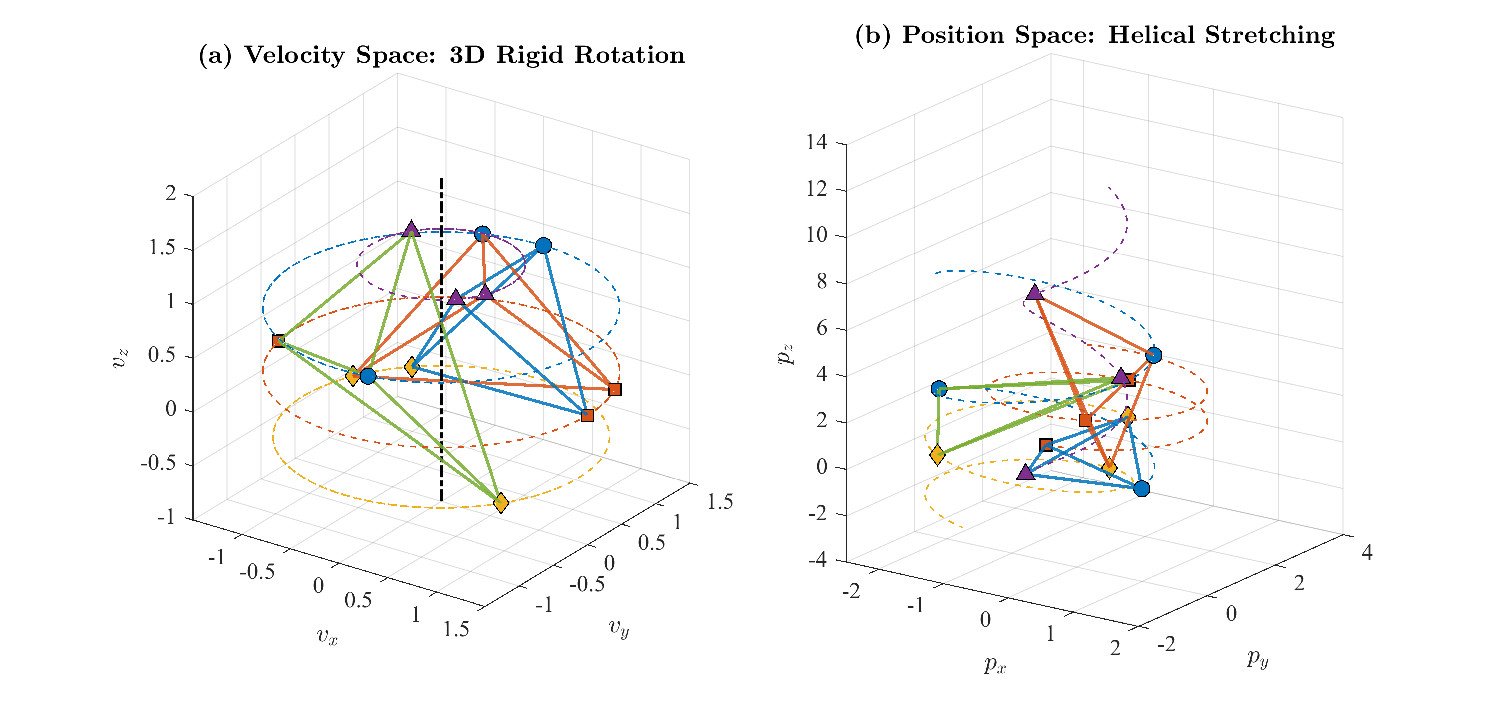}
		\vspace{-8mm}
		\caption{Conceptual schematics of the generic 3D rotational kinematic mapping. (a) In the velocity space, a 3D framework generically rotates around the singular axis $\text{Null}(\Omega)$ as a rigid body. (b) In the position space, constant but distinct null-space velocity components ($v_i^\parallel \neq v_j^\parallel$) cause varying helical pitches.}
		\label{fig:rotation_figures}    
	\end{figure}
	
	This mapping is visually validated in Fig. \ref{fig:rotation_figures} for a 3D network. As confirmed in Fig. \ref{fig:rotation_figures}(a), the 3D velocity framework uniformly rotates around $\text{Null}(\Omega)$ as a perfect rigid body. Crucially, the integration reveals profound geometric decoupling (Fig. \ref{fig:rotation_figures}(b)). Because the framework assigns distinct components $v_i^\parallel \neq v_j^\parallel$ to each agent, spatial trajectories naturally evolve into synchronized helices with varying ascending pitches. Consequently, the position-space framework autonomously undergoes continuous vertical stretching and shearing, powerfully visualizing how generalized velocity rigidity uniquely permits extreme spatial shape flexibility.
	
	\begin{corollary}[\textbf{Condition for Concentric Orbiting}]\label{cor:concentric}
		Under the conditions of Theorem \ref{thm:Rotational}, assuming initial velocities reside in the range space of $\Omega$, the network executes concentric circular orbits around a single common global center $c \in \mathbb{R}^d$ if and only if:
		\begin{equation}
			v_i(0) - v_j(0) = \Omega \big(p_i(0) - p_j(0)\big), \quad \forall (i, j) \in \mathcal{E}.
		\end{equation}
		Equivalently, this implies $v_i(0) = \Omega \big(p_i(0) - c\big)$ for all $i \in \mathcal{V}$.
	\end{corollary}
	\begin{proof}
		From Theorem \ref{thm:Rotational}, the individual center of rotation is $c_i = p_i(0) - \Omega^\dagger v_i(0)$. Macroscopic concentric orbiting requires $c_i = c_j = c$, yielding:
		\begin{equation}
			p_i(0) - \Omega^\dagger v_i(0) = p_j(0) - \Omega^\dagger v_j(0).
		\end{equation}
		Since $v_i(0) \in \text{Range}(\Omega)$, the pseudoinverse properties guarantee $\Omega \Omega^\dagger v_i(0) = v_i(0)$. Left-multiplying by $\Omega$ simplifies the relation to:
		\begin{equation}
			\Omega p_i(0) - v_i(0) = \Omega p_j(0) - v_j(0).
		\end{equation}
		Algebraic rearrangement immediately yields $v_i(0) - v_j(0) = \Omega \big(p_i(0) - p_j(0)\big)$. % \hfill $\blacksquare$
	\end{proof}
	To rigorously validate the geometric transition dictated by Corollary \ref{cor:concentric}, Fig. \ref{fig:concentric_comprehensive} presents a comprehensive comparative analysis of the rotational null space. Under generic, uncoupled initial states (Fig. \ref{fig:concentric_comprehensive}(a), (c), (e)), the position-velocity mapping yields independent orbiting centers, resulting in a macroscopic trajectory that lacks concentric synchronization. The velocity-position relationship exhibits random scatter (Fig. \ref{fig:concentric_comprehensive}(e)). 
	
	Conversely, enforcing the exact algebraic coupling condition $v_i = \Omega(p_i - c)$ (Fig. \ref{fig:concentric_comprehensive}(b), (d)) explicitly aligns the state vectors, forcing the network into perfectly synchronized concentric orbits. As visually proven by the center collapse analysis in Fig. \ref{fig:concentric_comprehensive}(f), this structural coupling algebraically forces the widely dispersed independent centers $c_i$ to collapse into a singular, unified geometric center $c$. This confirms that concentric orbiting is not an arbitrary maneuver, but a strict algebraic consequence of position-velocity coupling in the null space.
	\begin{figure}[htbp]
		\centering
		\includegraphics[width=\linewidth]{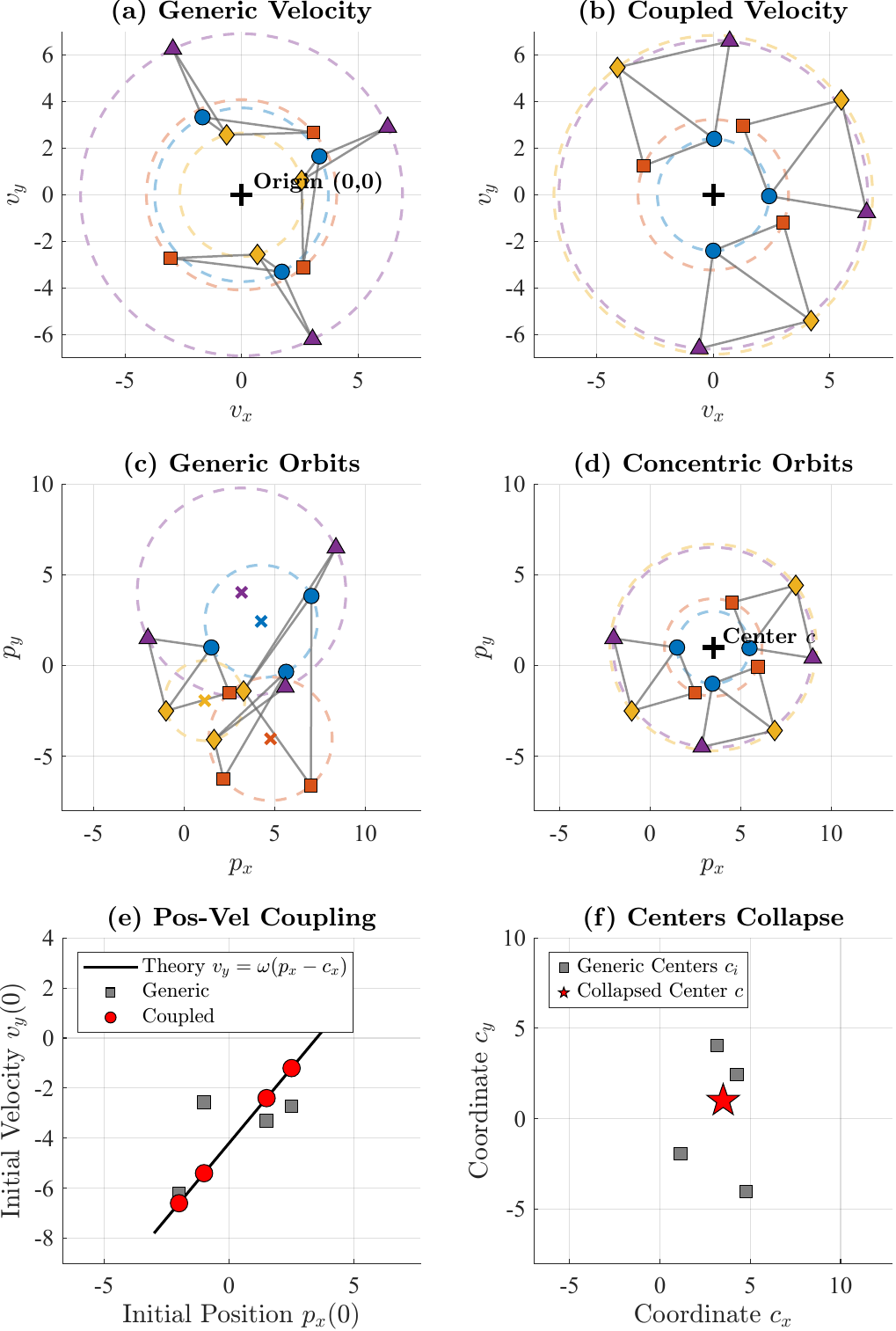}
		\caption{Comparative analysis of generic versus structurally coupled states in the rotational null space. \textbf{Left Column (Generic Case):} (a) Under algebraically independent initial states, the offset velocity framework revolves rigidly around the origin $(0,0)$. (c) Macroscopic integration yields individualized, non-concentric orbits around widely dispersed geometric centers $c_i$ (marked by crosses). (e) The generic uncoupled states (gray squares) exhibit random scatter, violating positional coupling. \textbf{Right Column (Coupled Special Case):} (b) The velocity graph rigidly orbits the origin. (d) Enforcing the exact algebraic condition $v_i = \Omega(p_i - c)$ explicitly offsets the formation and forces the network into perfectly synchronized concentric orbits around $c$. (f) Mathematical verification of the geometric transition: the widely dispersed independent centers algebraically collapse into the singular offset common center $c$ (the red star).}
		\label{fig:concentric_comprehensive}
		\vspace{-5 mm}
	\end{figure}
	
	\subsection{Scaling Subspace to Exponential Dilation}
	
	\begin{theorem}[\textbf{Exponential Dilation/Contraction}]\label{thm:Scaling}
		For a velocity graph utilizing generalized bearing constraints, if the acceleration profile resides strictly in the scaling null space $\mathcal{S}_S = \text{span}\{v(t)\}$, the multi-agent system undergoes an exponential spatial dilation or contraction while rigidly preserving its velocity-space topological directions.
	\end{theorem}
	\begin{proof}
		Let acceleration be $a_i(t) = s_0 v_i(t)$ ($s_0 \in \mathbb{R}$). The differential equation $\dot{v}_i(t) = s_0 v_i(t)$ yields $v_i(t) = e^{s_0 t} v_i(0)$. Substituting this into the normalized direction vectors yields:
		\begin{equation}
			\frac{e^{s_0 t}(v_i(0) - v_j(0))}{\|e^{s_0 t}(v_i(0) - v_j(0))\|} = \frac{v_i(0) - v_j(0)}{\|v_i(0) - v_j(0)\|},
		\end{equation}
		confirming that topological bearing constraints remain perfectly unbroken. Integrating the velocity vector gives the macroscopic position trajectory:
		\begin{equation}
			p_i(t) = p_i(0) + \int_{0}^{t} e^{s_0 \tau} v_i(0) \text{d}\tau = p_i(0) + \frac{e^{s_0 t} - 1}{s_0} v_i(0).
		\end{equation}
		For $s_0 > 0$, the distance grows exponentially along initial rays, manifesting as spatial dilation. % \hfill $\blacksquare$
	\end{proof}
	
	\begin{remark}
		In stark contrast to $\mathcal{S}_R$, where the generator $\Omega \in \mathfrak{so}(d)$ is dimension-dependent, the scaling generator relies on the isotropic scalar matrix $s_0 I_d$. Because $I_d$ is unconditionally full-rank in any dimension $d \ge 2$, the exponential dilation mapping is universally dimension-agnostic, generating straight outward rays without orthogonal drifting.
	\end{remark}
	
	\begin{corollary}[\textbf{Condition for Macroscopic Spatial Homothety}]\label{cor:scaling_homothety}
		Under generic uncoupled initial states, exponential dilation causes the physical spatial framework to distort and shear. The position-space framework will undergo a perfect spatial homothety about a common geometric center $c \in \mathbb{R}^d$, preserving its spatial similarity shape, if and only if:
		\begin{equation}
			v_i(0) - v_j(0) = \kappa \big(p_i(0) - p_j(0)\big), \quad \forall (i, j) \in \mathcal{E},
		\end{equation}
		equivalently implying $v_i(0) = \kappa (p_i(0) - c)$ for some $\kappa \neq 0$.
	\end{corollary}
	\begin{proof}
		Substituting $v_i(0) = \kappa(p_i(0) - c)$ into Theorem \ref{thm:Scaling} yields:
		\begin{align}
			p_i(t) &= p_i(0) + \frac{\kappa(e^{s_0 t} - 1)}{s_0} (p_i(0) - c).
		\end{align}
		Subtracting $c$ provides the relative geometry:
		\begin{equation}
			p_i(t) - c = \underbrace{\left[ 1 + \frac{\kappa}{s_0}(e^{s_0 t} - 1) \right]}_{\lambda(t)} (p_i(0) - c).
		\end{equation}
		The scalar multiplier is $\lambda(t)$. Consequently, inter-agent distances evolve identically as $\|p_i(t) - p_j(t)\| = |\lambda(t)| \|p_i(0) - p_j(0)\|$, mathematically proving that the network expands as a unified geometric entity relative to $c$. % \hfill $\blacksquare$
	\end{proof}
	
	The fundamental divergence between generic expansion and shape-preserving homothety (Corollary \ref{cor:scaling_homothety}) is visually deconstructed in Fig. \ref{fig:scaling_comprehensive}. In the generic case (Fig. \ref{fig:scaling_comprehensive}(a), (c)), uncoupled exponential dilation strictly preserves velocity-space bearing constraints but completely destroys the position-space spatial similarity, leading to severe structural distortion and shearing. Furthermore, as shown by the backward ray projection in Fig. \ref{fig:scaling_comprehensive}(f), the spatial rays of generic expansion fail to intersect, physically representing the lack of a shared homothety center.
	
	However, when the initial states are structurally coupled via $v_i(0) = \kappa(p_i(0)-c)$ (Fig. \ref{fig:scaling_comprehensive}(b), (d)), the expansion acts as a uniform spatial homothety. All backward rays perfectly intersect at the internal geometric center $c$ (Fig. \ref{fig:scaling_comprehensive}(f)), rigorously preserving the network's macroscopic geometric shape. This comparative verification powerfully underscores how specific state couplings can awaken absolute position-space geometric rigidity from the highly flexible velocity-rigid manifold.
	
	\begin{figure}[htbp]
		\centering
		\includegraphics[width=0.95\linewidth]{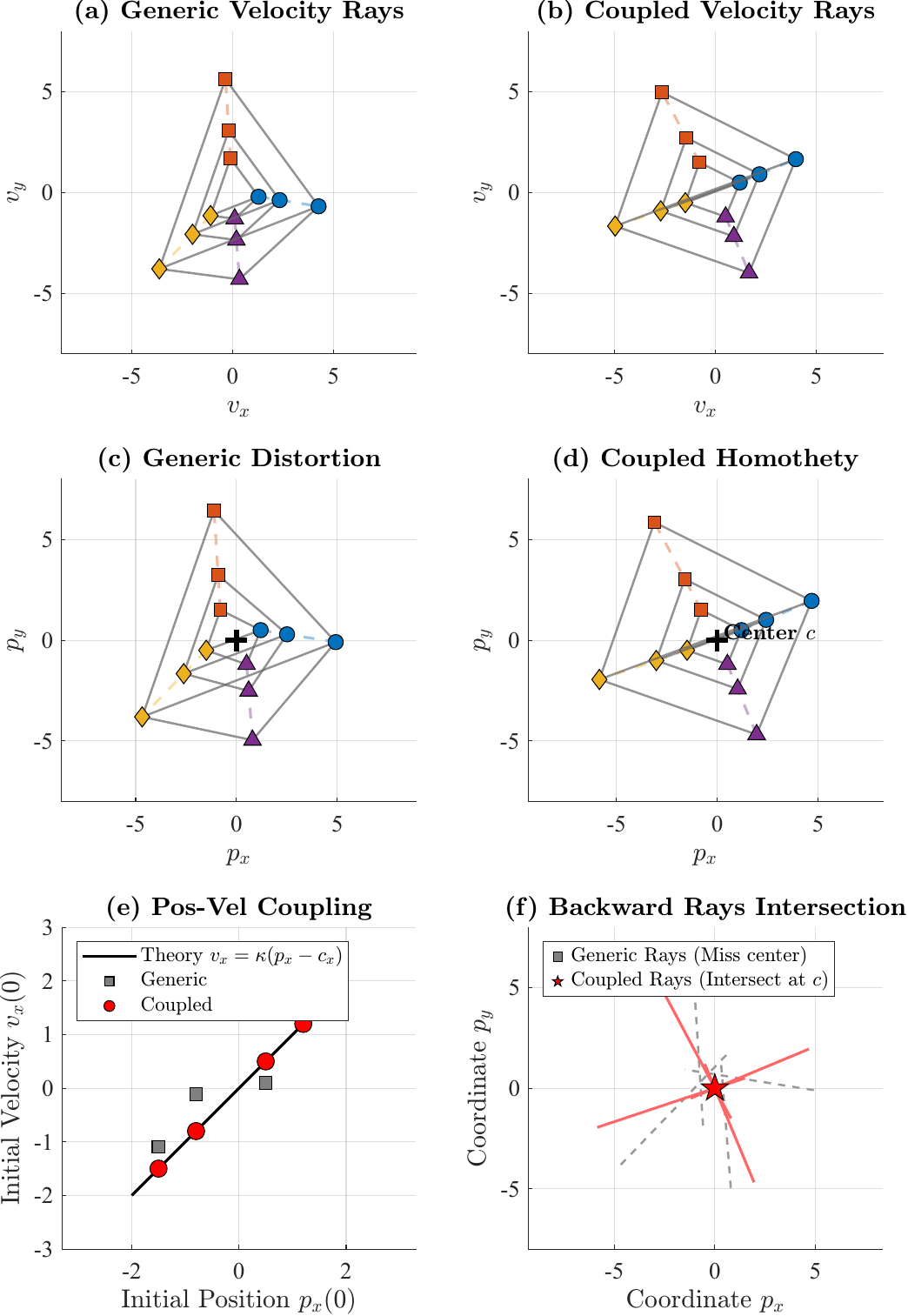}
        \vspace{-2mm}
		\caption{Comparative analysis of generic scaling expansion versus structurally coupled homothety. \textbf{Left Column (Generic Expansion):} (a) Driven by $s_0 I_d$, uncoupled velocity vectors expand exponentially along arbitrary rays. (c) As the network expands outward from the origin, it strictly preserves its velocity bearing constraints but undergoes severe spatial twisting and distortion. (e) Generic states randomly violate positional coupling. \textbf{Right Column (Coupled Homothety):} (b) Coupled initial velocities expand radially. (d) Enforcing $v_i(0) = \kappa(p_i(0)-c)$ perfectly aligns the network's expansion from an internal center $c$, resulting in a uniform spatial homothety that acts like a radially expanding balloon, rigidly preserving the spatial similarity shape. (f) Mathematical verification via backward ray projection: rays from the generic expansion (gray dashed) miss the origin, forming a disorganized web, whereas rays from the coupled expansion (red solid) perfectly intersect at the internal geometric homothety center $c$.}    
		\label{fig:scaling_comprehensive}
		\vspace{-6mm}
	\end{figure}
	
	%=========================================

	\section{DEGENERATE CONFIGURATIONS AND KINEMATIC SINGULARITIES}\label{sec:DEGENERATE}
	While classical positional rigidity requires frameworks to be strictly generic to avoid physical collisions, generalized velocity rigidity inherently embraces degenerate configurations. Coincident nodes ($v_i = v_j, i\neq j$) or collinear velocities represent physically meaningful kinematic states, such as consensus or unidirectional motion, rather than structural failures.
	
	\subsection{Velocity Coincidence and Topological Contraction}
	Since generalized velocity constraints rely on the relative velocity $v_i - v_j$, a singularity occurs when agents achieve identical velocities. Consider a distance constraint $r_k = \|v_i - v_j\|^2$; its Jacobian block is:
	\begin{equation}
		\frac{\partial r_k}{\partial v_i} = 2(v_i - v_j)^T.
	\end{equation}
	
	\begin{lemma}[\textbf{Rank Drop under Velocity Consensus}]
		If connected agents in $\mathcal{G}$ achieve velocity consensus ($v_i = v_j$), the corresponding row blocks in $R_v(v)$ evaluate to exactly zero, resulting in a local rank drop of the rigidity matrix.
	\end{lemma}
	\begin{proof}
		If $v_i = v_j$, then $(v_i - v_j)^T = \mathbf{0}^T$. The partial derivatives with respect to $v_i$ and $v_j$ vanish, annihilating the constraints imposed by edge $(i,j)$. % \hfill $\blacksquare$
	\end{proof}
	
	Physically, this rank drop indicates a topological contraction where agents $i$ and $j$ coalesce into a single rigid body within the velocity manifold. The relative acceleration is constrained to zero ($\dot{v}_i = \dot{v}_j$). As illustrated in Fig. \ref{fig:translation}(a), global velocity consensus collapses the network into a single point in the velocity graph. Consequently, the network translates cohesively as a rigid body in the physical position space (Fig. 1(b)). It must be strictly emphasized that a topological contraction in the velocity space ($v_i = v_j$) merely dictates parallel curved trajectories, inherently guaranteeing zero spatial collisions provided the initial physical positions are distinct ($p_i(0) \neq p_j(0)$). This specific state characterizes the current kinematic configuration, which is mathematically distinct from the translational trivial motion subspace $\mathcal{S}_T$ governing uniform acceleration (Section IV.A).
	
	\begin{figure}[htbp]
		\centering
		\includegraphics[width=\linewidth]{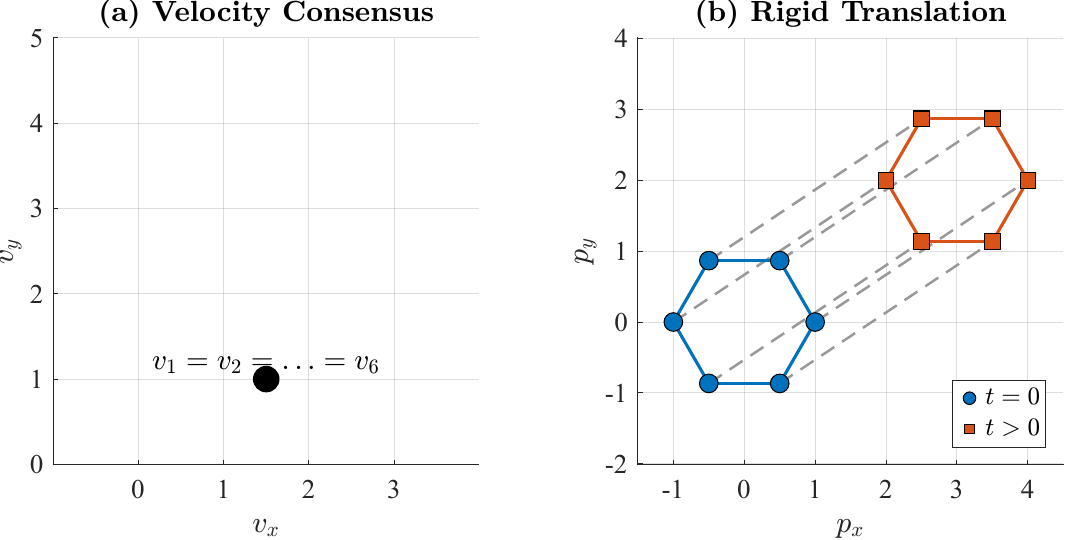}
		\vspace{-8mm}
		\caption{Constant-velocity translation illustrating velocity consensus. \textbf{(a)} In the velocity space, all agents coalesce into a single 0-dimensional point ($v_i = v_j$), resulting in a local rank drop. \textbf{(b)} In the physical position space, this topological contraction naturally maps to a cohesive, absolutely rigid spatial translation, smoothly transporting the initial framework along parallel linear trajectories without collisions.}
		\label{fig:translation}
		\vspace{-5mm}
	\end{figure}
	
	\subsection{Velocity Collinearity and Subspace Confinement}
	\begin{theorem}[\textbf{Degeneration of the Rotational Null Space}]
		If the network reaches a collinear velocity state, such that all $v_i(t)$ span a 1-dimensional subspace in $\mathbb{R}^d$ ($d \ge 2$), the rotational trivial motion subspace $\mathcal{S}_R$ degenerates, confining the macroscopic synchronized circular orbiting to planar oscillatory or linear sweeping motions.
	\end{theorem}
	\begin{proof}
		Let the initial velocities be linearly dependent: $v_i(0) = \alpha_i \mathbf{w}$, where $\alpha_i \in \mathbb{R}$ and $\mathbf{w} \in \mathbb{R}^d$ is a common unit vector. Under the rotational generator $\Omega \in \mathfrak{so}(d)$, the velocity trajectory is:
		\begin{equation}
			v_i(t) = e^{\Omega t} v_i(0) = \alpha_i \big(e^{\Omega t} \mathbf{w}\big).
		\end{equation}
		To determine the macroscopic position trajectories, we orthogonally decompose the unit vector as $\mathbf{w} = \mathbf{w}^\perp + \mathbf{w}^\parallel$, where $\mathbf{w}^\perp \in \text{Range}(\Omega)$ and $\mathbf{w}^\parallel \in \text{Null}(\Omega)$. Integrating the velocity yields:
		\begin{align}
			\hspace{-1mm}
			p_i(t) &= p_i(0) + \alpha_i \int_{0}^{t} e^{\Omega \tau} (\mathbf{w}^\perp + \mathbf{w}^\parallel) \text{d}\tau \nonumber \\
			&= p_i(0) + \alpha_i \underbrace{\Omega^\dagger (e^{\Omega t} - I_d) \mathbf{w}^\perp}_{\text{Planar oscillatory}} + \alpha_i \underbrace{\mathbf{w}^\parallel t}_{\text{Linear sweeping}}.
		\end{align}
		Let $\mathbf{s}(t) = \Omega^\dagger (e^{\Omega t} - I_d) \mathbf{w}^\perp + \mathbf{w}^\parallel t$ denote the shared 1-dimensional trajectory basis. Because the action of $\Omega$ restricts $e^{\Omega t} \mathbf{w}^\perp$ to a 2D invariant plane, its integral mathematically forms a bounded planar curve (oscillation), while the null space component introduces a constant linear sweeping drift. Consequently, the position of each agent is strictly confined to $p_i(t) = p_i(0) + \alpha_i \mathbf{s}(t)$. 
		
		The relative distance vector between any pair evolves as $p_i(t) - p_j(t) = (p_i(0) - p_j(0)) + (\alpha_i - \alpha_j) \mathbf{s}(t)$. This proves that the framework cannot sustain multidimensional rigid orbiting; instead, it undergoes a linear shearing deformation along the restricted planar oscillatory or sweeping paths. % \hfill $\blacksquare$
	\end{proof}
	% \begin{theorem}[\textbf{Degeneration of the Rotational Null Space}]
		%      If the network reaches a collinear velocity state, such that all $v_i(t)$ span a 1-dimensional subspace in $\mathbb{R}^d$ ($d \ge 2$), the rotational trivial motion subspace $\mathcal{S}_R$ degenerates, confining the macroscopic synchronized circular orbiting to planar oscillatory or linear sweeping motions.
		% \end{theorem}
	% \begin{proof}
		%     Let all velocities be linearly dependent: $v_i = \alpha_i \mathbf{w}$ for $\alpha_i \in \mathbb{R}$ and unit vector $\mathbf{w} \in \mathbb{R}^d$. For any skew-symmetric $\Omega \in \mathfrak{so}(d)$, the rotational generator is:
		%     \begin{equation}
			%         \delta v_i = \Omega (\alpha_i \mathbf{w}) = \alpha_i (\Omega \mathbf{w}).
			%     \end{equation}
		%     Because velocities lack multi-dimensional span, the Lie algebra action is restricted to the orthogonal complement of $\mathbf{w}$. Consequently, the full rank of the rotational null space cannot be exercised, mapping $e^{\Omega t} v_i(0)$ onto a lower-dimensional manifold. % \hfill $\blacksquare$
		% \end{proof}
	
	This subspace confinement is visually demonstrated by the shear motion in Fig. \ref{fig:shear}. Collinear velocity vectors collapse the velocity graph onto a 1-dimensional axis (Fig. \ref{fig:shear}(a)). This effectively degenerates the rotational degrees of freedom, enforcing a linear, shear-like deformation in the macroscopic position space (Fig. \ref{fig:shear}(b)).
	
	\begin{figure}[h!]
		\centering
		\hspace{-10mm}
		\includegraphics[width=\linewidth]{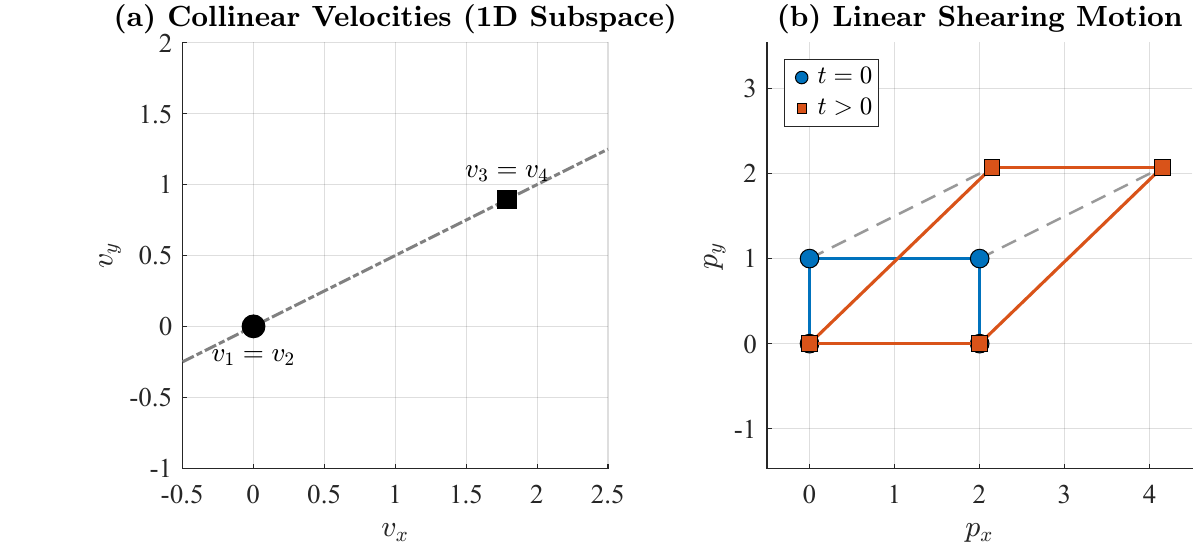}
		\vspace{-3mm}
		\caption{Illustrations of degenerate velocity configurations and subspace confinement. \textbf{(a)} In the velocity space, agents exhibit collinear velocities, representing a 1D subspace confinement where the graph collapses onto a single axis. \textbf{(b)} In the position space, this degeneration of rotational degrees of freedom mathematically enforces a linear, shear-like deformation, transforming the initial rectangular framework into a sheared parallelogram.}
		\label{fig:shear}
		\vspace{-4mm}
	\end{figure}
	%====================================================
	\section{NUMERICAL SIMULATIONS}
	\label{sec:simulations}
	
	To validate the exact kinematic mappings and demonstrate the extreme spatial flexibility endowed by generalized velocity rigidity, we simulate a continuous, multi-phase macroscopic maneuver using a planar network of $n=4$ agents. The network seamlessly alternates between degenerate velocity singularities and fully actuated Lie group manifolds, all while rigorously preserving the generalized velocity constraints. 

    %{\color[HTML]{2da44e} Throughout the simulations, we assume the framework maintains infinitesimal velocity rigidity at all times. Under this assumption, the constraint residual vanishes identically, so the differential inclusion of Lemma~\ref{lem-diff-inclusion} holds at every instant.}
	
	To ensure visual clarity and prevent trajectory overlapping, a continuous macroscopic translational drift is systematically integrated into the sequence. This elegantly unrolls the position-space trajectory horizontally, akin to a timeline (Fig. \ref{fig:multiphase}, Top). The corresponding instantaneous velocity-space frameworks, along with the trajectory of the velocity centroid $v_c(t) = \frac{1}{n}\sum v_i(t)$, are captured in the bottom panels (Fig. \ref{fig:multiphase}(a)-(d)). Tracking $v_c(t)$ explicitly visualizes the algebraic decoupling between the macroscopic translation and the internal structural morphing.
	
	\begin{figure*}[ht!]
		\centering
		% \hspace{-20mm}
		\includegraphics[width=1\linewidth]{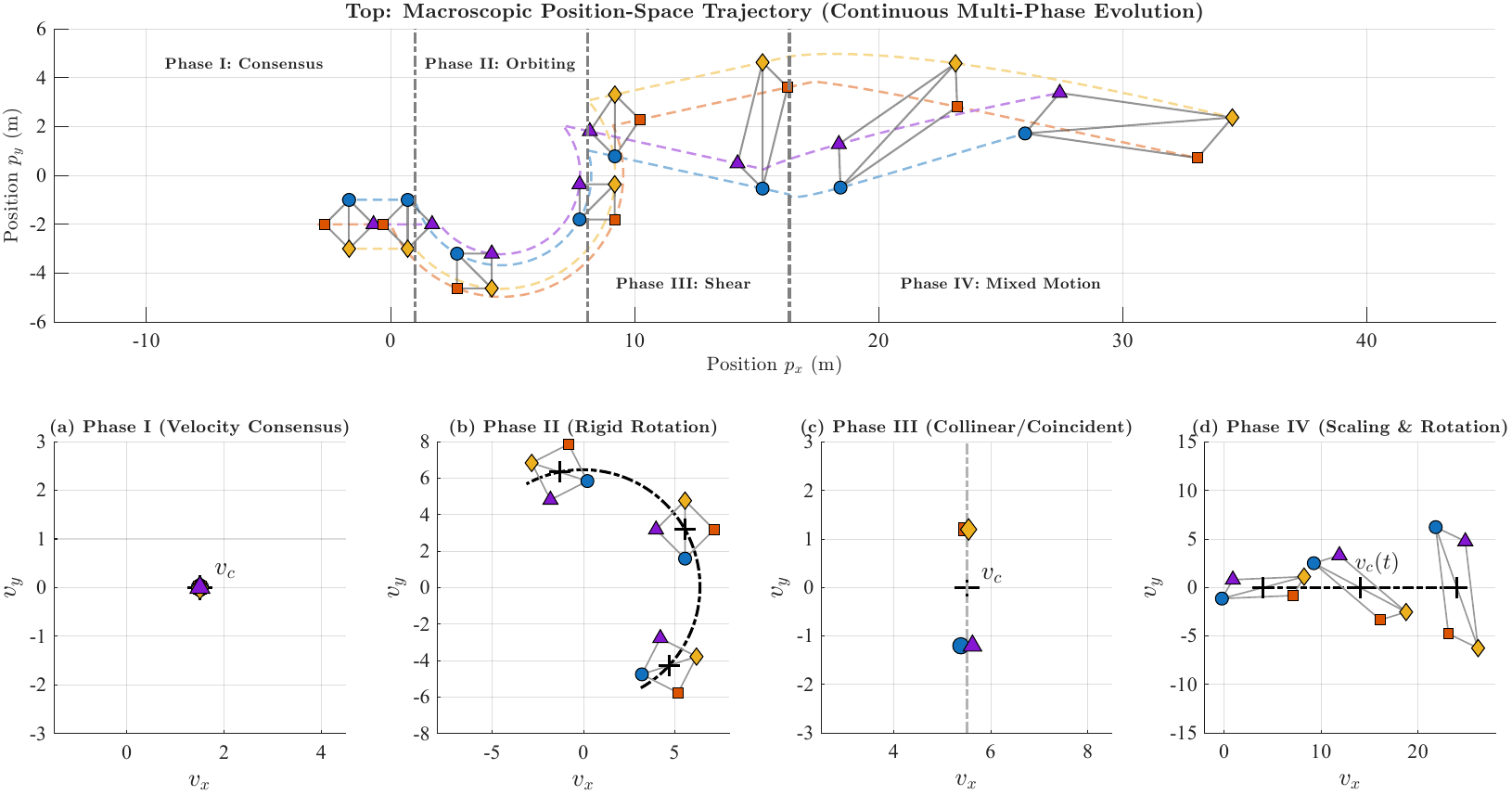}
		\vspace{-8mm}
		\caption{Continuous multi-phase simulation of the generalized velocity-rigid network. \textbf{Top:} Unrolled macroscopic trajectory in the position space. The network transitions through distinct kinematic modes, illustrating extreme topological morphing capabilities. \textbf{Bottom (a)-(d):} Corresponding instantaneous states in the velocity space. All subplots maintain identical aspect ratios. The black cross ($+$) and dash-dotted lines track the instantaneous velocity centroid $v_c(t)$. (a) Phase I exhibits a degenerate consensus point. (b) Phase II displays a rotating rigid body, where $v_c(t)$ traces a circular arc. (c) Phase III highlights subspace confinement, collapsing the graph onto a 1D line to induce spatial shearing. (d) Phase IV exhibits a fully actuated, expanding, and rotating framework superimposed on a translational drift.}
		\label{fig:multiphase}
		\vspace{-5mm}
	\end{figure*}

	\textbf{Phase I: Velocity Consensus (Rigid Translation) ($t \in [0, 2)$).} 
	The network initiates with a global velocity consensus ($v_i = v_j$). As analyzed in Section III.A, the velocity framework degenerates into a single zero-dimensional point, with the centroid $v_c$ remaining static (Fig. \ref{fig:multiphase}(a)). In the position space, this topological contraction naturally maps to a cohesive, rigid spatial translation, smoothly transporting the diamond framework forward.
	
	\textbf{Phase II: Concentric Orbiting (Rigid Rotation) ($t \in [2, 4)$).} 
	The network seamlessly switches to the rotational null space $\mathcal{S}_R$. By strictly coupling velocities to positions via $v_i = \Omega(p_i - c)$, the velocity framework transforms into a rotating rigid body. Notably, as the position-space formation traces a concentric orbit around the local center $c$, its velocity centroid $v_c(t)$ autonomously traces a perfect circular arc around the origin (Fig. \ref{fig:multiphase}(b)).
	
	\textbf{Phase III: Shear Motion via Collinearity ($t \in [4, 5.5)$).} 
	To navigate a hypothetical constrained environment, the network exploits a velocity-space singularity. Velocities are commanded to a 1D collinear subspace with coincident nodes, superposed with a common horizontal drift $c(t)$. This collapses the velocity graph into a discrete, vertically aligned line segment shifting horizontally (Fig. \ref{fig:multiphase}(c)). This specific degeneracy forces the position-space framework to stretch vertically while flying forward, seamlessly elongating into a narrow geometry.
	
	\textbf{Phase IV: Mixed Motion (Scaling \& Rotation) ($t \in [5.5, 6.5]$).} 
	In the final phase, the acceleration profile orchestrates a superposition of scaling, rotation, and an accelerating horizontal translation. Driven by the composite generator $v_i(t) = s_0(p_i - c_m(t)) + \Omega_m(p_i - c_m(t)) + v_T(t)$, the velocity framework simultaneously expands and rotates around its continuously shifting centroid $v_c(t)$ (Fig. \ref{fig:multiphase}(d)). By dynamically synchronizing the geometric center $c_m(t)$ with the collective drift, the position-space network undergoes a continuous spatial homothety and rotational sweep, manifesting as a forward-moving macroscopic vortex. 
	
	Throughout all four phases—despite undergoing severe spatial shearing, orbital flipping, and exponential dilation—the fundamental algebraic properties of the velocity graph remain flawlessly intact. This comprehensive simulation forcefully corroborates the core paradigm: generalized velocity rigidity securely locks macroscopic continuous dynamics while empowering the swarm with unparalleled spatial shape flexibility.

	\section{CONCLUSION}
	\label{sec:conclusion}
	This paper developed a generalized velocity-rigidity framework that extends graph rigidity from the classical position space to the continuous-time velocity space. Rather than constraining a formation to maintain a prescribed spatial shape, the proposed framework characterizes how relative velocity constraints can coordinate coherent network-level motion. By analyzing the null space of the generalized velocity rigidity matrix, we showed that its trivial motions correspond to several fundamental collective behaviors, including synchronized translation, rotational motion associated with $SO(d)$, and exponential homothetic scaling.

    A key feature of generalized velocity rigidity is that it constrains instantaneous kinematic relations while allowing substantial flexibility in the corresponding position configuration. Velocity-degenerate configurations, such as consensus and collinearity, can therefore serve as transitional modes for coordinated reconfiguration. The simulations demonstrated shearing, stretching, contraction, and passage through narrow regions while preserving the prescribed velocity-level constraints.

	This flexibility, however, also exposes several limitations of the present formulation. Because the position configuration is regulated only indirectly through velocity-level constraints, generalized velocity rigidity alone does not guarantee bounded formation size or convergence to a prescribed spatial geometry. Consequently, the inter-agent distances or overall formation diameter may grow beyond desirable limits unless additional position- or distance-dependent constraints are introduced. Moreover, collision avoidance is not explicitly encoded in the generalized velocity-rigidity constraints. Additional safety mechanisms are therefore required when the agents operate in dense, cluttered, or dynamically constrained environments.
    
    %{\color[HTML]{2da44e}Nonetheless, this flexibility also reveals an important limitation. In the proposed framework, the formation position is regulated only indirectly through velocity-level constraints, so the resulting spatial extent may grow without bounds unless additional regulation is imposed. Moreover, collision avoidance between agents is not explicitly considered by the velocity rigidity, which means that safety constraints must be incorporated separately when the swarm operates in dense or cluttered environments.}
	
	Future work will investigate velocity-rigidity-based motion planning and distributed control architectures, including directed leader--follower schemes. Extensions to safety-constrained coordination, robustness, and higher-order or nonholonomic agent dynamics will also be considered.
    
    %The framework developed in this paper also opens several directions for future research. One direction is motion and path planning, where generalized velocity-rigid motions and their singular configurations may serve as structured motion primitives for coordinated navigation, large-scale reconfiguration, and passage through constrained environments. A second direction is the development of distributed control architectures, including leader--follower and directed-interaction schemes, for regulating or switching among desired velocity-rigid motion patterns using only local relative information. Further extensions may incorporate geometric and safety constraints, such as formation-diameter bounds, inter-agent collision avoidance, and obstacle avoidance, together with robustness against disturbances, sensing uncertainty, and model mismatch. Extending the theory to higher-order agent dynamics and nonholonomic systems would further clarify the role of generalized velocity rigidity in practical multi-agent motion coordination.
	
	\bibliographystyle{IEEEtran}
	\bibliography{References}
	
\end{document}